\documentclass{ifacconf}

\usepackage{graphicx,xcolor}      
\usepackage{natbib}        
\usepackage{amsmath,amsfonts}
\usepackage{enumitem} 

\usepackage{subcaption}

\begin{document}
\begin{frontmatter}

\title{Optimum adaptation of a Steiner network\thanksref{footnoteinfo}} 

\thanks[footnoteinfo]{M. Ye was supported by the Australian Government through the Australian Research Council (DE250100199).}

\author[First]{Manou Rosenberg} 
\author[Second]{Mengbin Ye} 
\author[Third]{Brian D.O. Anderson}

\address[First]{Centre for Optimisation and Decision Science, Curtin University, Perth, Australia (e-mail: manou.rosenberg@curtin.edu.au).}
\address[Second]{Adelaide Data Science Centre, Adelaide University, Adelaide, Australia (e-mail: ben.ye@adelaide.edu.au)}
\address[Third]{School of Engineering, Australian National University, Canberra, Australia (e-mail: brian.anderson@anu.edu.au)}

\begin{abstract}                
The Euclidean Steiner tree problem, normally posed in two dimensions,  seeks to connect a set of prescribed terminal nodes by placing additional nodes, known as Steiner points, with edges connecting such nodes either to another Steiner point or a terminal node, and  with the placements minimising the sum of all the edge lengths of the associated tree. We consider a problem in which we start with a known solution to a Steiner tree problem, and the terminal positions are then perturbed. A first-order approximation theorem is established for efficiently updating the Steiner point positions to recover a Steiner tree solution after the perturbations to terminal nodes. Numerical examples illustrate the effectiveness of our approach (including  a stepwise application for large perturbations) as well as its limitations. 
\end{abstract}

\begin{keyword}
Large-scale complex systems, Large-scale and networked optimization problems,
Steiner tree, Network Optimisation, Perturbed Networks
\end{keyword}

\end{frontmatter}

\section{Introduction}\label{sec:intro}
The Euclidean Steiner tree (EST) problem, in its classical form, seeks the shortest possible network interconnecting a given finite set of points in $\mathbb{R}^2$ (termed \textit{terminals}), allowing for the inclusion of additional points (termed \textit{Steiner} points) at optimal locations to reduce its total length. (The network length is defined as the sum of the edge lengths, i.e. the sum of the distances between every connected pair of points). Historically, this optimisation problem has its roots in the early 19th century and was proposed by~\cite{Gergonne1811}. In the Euclidean version the total length is measured by the sum of Euclidean edge lengths and a solution to such a problem instance, i.e. the set of terminals and their locations, along with the edge set, is called Euclidean Steiner Minimum Tree (ESMT). While the problem has been shown to be NP-hard \citep{Garey1977}, it has garnered significant interest owing to its mathematical structure and its applicability across a range of domains. Classical results, such as Gilbert and Pollak’s conjecture \citep{Gilbert1968} on the ratio of lengths of Steiner trees to minimum spanning trees, and early exact algorithms based on geometric constructions and branch-and-bound approaches \citep{hwang1992steiner}, have laid the foundation for both theoretical study and algorithmic development.  Several heuristics and approximation algorithms have been developed to find near-optimal solutions efficiently (see e.g.~\cite{zachariasen1999concatenation}).

ESMTs have extensive real-world relevance and can be applied to various network design problems. Classical applications arise in circuit layout, wireless communication, and transportation planning. More recently, ESMTs have been used in the planning, design and analysis processes of emerging domains  that involve networked systems such as  electricity network design \citep{rosenberg2022finding}, microprocessor layout \citep{lin2019construction}, modelling biomolecular structures, \cite{mondaini2008steiner}), hydrogen pipeline networks \citep{Hammond2025}, and the planning of robust and cost-efficient infrastructure networks in urban environments \citep{laha2023steiner}. 
 In many real-world scenarios an optimal network design needs to respect further constraints. Recent research has therefore expanded both problem formulations and solution methods for the Steiner problem. In obstacle-aware problems, the connections between points must navigate around inaccessible regions (hard obstacles) or traverse regions with increased costs per unit length (soft obstacles). \cite{garrote2019weighted} developed a deterministic heuristic approach. Their algorithm is applied to telecommunication networks, where soft obstacles represent disaster-prone areas that are exposed to a higher risk of natural disasters such as hurricanes, floods or earthquakes. \cite{rosenberg2021genetic} developed a genetic algorithm approach for the Steiner tree problem with solid and soft obstacles.  Another area of application is the planning of electricity networks, for example. In this case, soft obstacles could represent higher risk or altitude zones, private properties, difficult terrain, lakes, or protected areas~\citep{rosenberg2022finding}.

In this paper, we explore the methodology of recovering a  Euclidean Steiner minimum tree (without obstacles) after perturbation in the terminal nodes. The motivation for considering perturbations is twofold. First, we seek  to develop a modelling framework that considers practical conditions where input points are uncertain or where network requirements change after finding an initial minimal-length network design. In urban communication networks, the locations of service points, e.g., utility stations, routers, data centers, might change due to new regulations, construction projects, or better land use. In a power distribution network,  the location of transformers or other equipment might shift upon deployment due to relocations or upgrades. In sensor networks used for environmental monitoring or industrial automation, the locations of sensors (e.g., temperature sensors, motion detectors) might require adjustment  to accommodate operational needs, equipment updates, or to improve coverage.
The second aim is to establish a theoretical foundation and demonstrate its effectiveness on problems without obstacles, with the objective of extending it to the substantially more challenging scenarios involving soft obstacles, where recalculating Steiner minimum tree solutions is computationally challenging. In addition to the terminal nodes and the Steiner nodes, the network's entry/exit points at soft obstacle boundaries constitute a further class of points; we conjecture that the positions of these points, after perturbation of the terminal nodes, can be determined by extending the methods of this paper. 

Despite their practical importance, systematic approaches to the perturbed ESMT problem are scarce. The existing literature predominantly focuses on classical or static problem instances, leaving a gap in the study of network optimality under perturbations. To the best of the authors' knowledge, there is no direct literature on the perturbed Euclidean Steiner tree problem as described in this paper. However, there are some related works that consider the insertion and deletion of terminal nodes or consider approximation algorithms in general. \cite{chapeau} proposed an approach that incrementally updates the tree structure when adding or removing nodes, maintaining a near-optimal solution without requiring full recomputation. \cite{chan2024fully} presented a dynamic probabilistic approximation algorithm for the Euclidean Steiner tree problem in $d$ dimensions, which is capable of maintaining a near-optimal $(1 + \varepsilon)$-approximate solution as terminals are inserted or deleted. \cite{raikwar2025dynamic} address the Steiner tree problem in graphs and propose a method to update the Steiner Minimum Tree in graphs after changes to the underlying nodes and edges.

The key contributions of this paper are the formulation of the Euclidean Steiner tree problem with perturbed terminal nodes, and the development of a first-order approximation theorem for recovery of a Steiner tree solution after perturbation. After small changes in the position of some or all terminals, the proposed approach approximates the updated positions for the Steiner points. This method requires the assumption that the  perturbation is sufficiently small such that the topology of the Steiner minimum tree does not change.
The method requires a Hessian matrix to be invertible, and we analytically prove that it is always invertible. This gives a theoretical foundation for the proposed approximated Steiner tree solution, and going further we present empirical examples to test the stepwise application of the approximation theorem to a large terminal node position change.

The remainder of this paper is organised as follows. Section 2 introduces mathematical preliminaries and recent literature on ESMTs. Section 3 presents the main results including the first-order approximation theorem for the perturbed ESMT. Section 4 gives numerical examples showing the application of the approximation theorem and its limitations, and Section 5 concludes the paper and gives an outlook on future contributions, for which this work is the precursor.

\section{Background and Problem Formulation}\label{sec:background}

In this section, we introduce the classical Euclidean Steiner Tree (EST) problem and recall known properties of a solution to this problem, termed a Euclidean Steiner minimum tree (ESMT). Then, we formulate our problem of interest, which involves a perturbation\footnote{While \textit{perturbation} typically refers to a small change, in this paper we will use this term to also denote a larger change in the terminals that can be achieved by a sequence of small changes.} to the nodes of the EST problem, and provide motivating arguments to its importance.

\subsection{Euclidean Steiner Minimum Trees}

In this paper, we consider graphs $\mathcal T = (\mathcal V,\mathcal E)$ that are embedded in the plane, i.e. $\mathbb R^2$. Thus, for the set of nodes $\mathcal V = \{v_1, \hdots, v_m\}$, we associate with node $v_i$ the position vector $p_{v_i} = [x_{v_i}, y_{v_i}]^\top$. An edge $(i,j) \in \mathcal E$ has the associated Euclidean distance $\Vert p_{v_i} - p_{v_j}\Vert$.

Let $T = \{t_1, \hdots, t_n\}$ be a set of $n\geq 2$ nodes, which we call \textit{terminals}. Consistent with the above, each node $t_i$ is located in $\mathbb R^2$, and $p_{t_i} = [x_{t_i}, y_{t_i}]^\top$ denotes its position.

\begin{prob}[Euclidean Steiner Tree problem]\label{prob:EST}
    For a given set of terminals $T$, find a connected graph $\mathcal T = (\mathcal V, \mathcal E)$ such that i) $T\subseteq \mathcal V$, and ii) the total distance --obtained by summing the distance over the edges-- is minimised.
\end{prob}

The resulting network $\mathcal T$ is in fact a tree, and may contain additional nodes, i.e. $\mathcal V = T \cup S$. In more detail, the set of nodes $S = \{s_1,\hdots, s_k\}$ are called \textit{Steiner} nodes, with $s_j$ having the position vector $p_{s_j} = [x_{s_j}, y_{s_j}]^\top$. Given a set of terminals $T$, there can be multiple  graphs that solve Problem~\ref{prob:EST}, and each graph is called a \textit{Euclidean Steiner minimal tree} (ESMT). Any ESMT $\mathcal T$ has the following properties:
\begin{enumerate}[label=P\arabic*]
    \item There are at most $n-2$ Steiner points. That is, $k\leq n-2$.\label{numbercon}
    \item\label{degreecon}  Each terminal is connected to at most three other nodes and each Steiner point is  adjacent to exactly three other nodes (\textit{degree condition}).
    \item\label{anglecon}  Each pair of edges in the tree meets at an angle of at least $2\pi /3$ and the three edges at a Steiner point meet exactly at an angle of $2\pi /3$ (\textit{angle condition}).
\end{enumerate}
For convenience, we partition $\mathcal{E}$ into edges connecting terminal nodes to terminal nodes $\mathcal{E}_{T}$, edges connecting terminals and Steiner nodes $\mathcal{E}_{T,S}$, and those connecting Steiner nodes to Steiner nodes  $\mathcal{E}_{S}$. Note that any of the aforementioned edge sets can be empty depending on the EMST solution, and there holds $\mathcal{E} = \mathcal{E}_{T}\cup \mathcal{E}_{T,S}\cup \mathcal{E}_{S}$.

As noted in Section~\ref{sec:intro}, the EST problem is proven to be NP-hard \citep{Garey1977}, and thus heuristic algorithms and approximation schemes generally provide the best method for identifying candidate ESMTs when there are a large number of nodes.  
Worst-case time complexity for exact algorithms solving the Euclidean Steiner Tree Problem is exponential, typically $\mathcal{O}(2^n)$, where $n$ is the number of terminals \citep{hwang1992steiner}.
 For specific circumstances, where the number of Steiner points and the edge connections are known, \cite{melzak1961problem} and \cite{HWANG1986235} have developed an exact algorithm for the EST problem that runs in linear time. The proposed approximation theorem of this paper works under the same assumptions, but importantly, our approach can be extended to more complex cases including obstacle-aware Steiner tree problems.

\subsection{Perturbed Euclidean Steiner Tree Problem}\label{ssec:problem_definition}

Here, we shall define a \textit{perturbed Euclidean Steiner Tree problem} (perturbed EST problem), which is the main focus of this paper. 

Roughly speaking, we are interested in the scenario in which we are given a set of terminals $T$, and we have already obtained an ESMT $\mathcal{T} = (T\cup S, \mathcal E)$ that solves the EST problem for $T$, using a set of Steiner points $S$. Then, we suppose the terminal positions are perturbed by a small amount, and we seek an algorithm that can rapidly compute the updated locations of the Steiner points $S$ \textit{without having to re-solve from scratch} the EST problem using the perturbed terminal positions. 

 For a given set of terminal nodes $T$, let $S$ be a set of Steiner points such that $\mathcal{T} = (T\cup S,\mathcal E)$ is an ESMT that solves Problem~\ref{prob:EST}. Let us define the column vectors $s\in \mathbb R^{2k}$ and $t \in\mathbb R^{2n}$ for the Steiner node locations and terminal node locations, respectively, as 
\begin{subequations}
\begin{align}
    s & = \begin{bmatrix}\label{eq:s_vector}
        x_{s_1} & y_{s_1} & x_{s_2} & y_{s_2} & \cdots & x_{s_k} & y_{s_k} 
    \end{bmatrix}^\top, \\
    t & = \begin{bmatrix}
        x_{t_1} & y_{t_1} & x_{t_2} & y_{t_2} & \cdots & x_{t_n} & y_{t_n} 
    \end{bmatrix}^\top. \label{eq:t_vector}
\end{align}
\end{subequations}
Note $s= [p_{s_1}^\top \ p_{s_2}^\top \ \cdots \ p_{s_k}^\top]^\top$ and $t = [p_{t_1}^\top \ p_{t_2}^\top \  \cdots \ p_{t_n}^\top]^\top$. Let
\begin{equation}\label{eq:delta_t}
    \delta t = \begin{bmatrix}
        \delta x_{t_1} & \delta y_{t_1} &\delta x_{t_2} & \delta y_{t_2} & \cdots & \delta x_{t_n} & \delta y_{t_n} 
    \end{bmatrix}^\top,
\end{equation}
be a perturbation vector, and suppose the terminals are perturbed so that the new terminal node positions are:
\begin{equation}
     t + \delta t .
\end{equation}

We now want to find a perturbed ESMT solution ${\bar{\mathcal{T}}}_{t + \delta t}$. In other words, we seek new locations for the Steiner nodes 
\begin{equation}
    s + \delta s,
\end{equation}
given the shifted terminal positions $t +\delta t$. Here,
\begin{equation}\label{eq:delta_s}
    \delta s = \begin{bmatrix}
        \delta x_{s_1} & \delta y_{s_1} & \delta x_{s_2} & \delta y_{s_2} & \cdots & \delta x_{s_k} & \delta y_{s_k} 
    \end{bmatrix}^\top,
\end{equation}
gives the shift in the Steiner node locations. More formally, we will define the problem as follows.

\begin{prob}[Perturbed Euclidean Steiner Tree Problem]\label{prob:no_obstabcles}
    With notation as above, let $S$ be a set of Steiner nodes which, given a set of terminal nodes $T$, solves the EST problem by forming the ESMT $\mathcal T$. Suppose the terminal nodes are shifted by the perturbation vector $\delta t$, to take new positions $t + \delta t$. Find the perturbation vector $\delta s$ such that the tree ${\bar{\mathcal{T}}}_{t + \delta t}$ is a perturbed ESMT with the terminal node positions $t + \delta t$ and Steiner node positions $s+ \delta s$ . In particular, find a matrix $X$ such that $\delta s=X\delta t+o(\| \delta t\|)$.
\end{prob}

 For a small perturbation, of course one can simply take $\delta s=X\delta t$. For a large perturbation, one can break up the perturbation $\delta t$ into a sum of smaller perturbations $\delta t_i, i=1,2,\dots,K$ and apply the first order approximation formula with $X$ determined successively by $t, t+\delta t_1, t+\delta t_1+\delta t_2,\dots$. This is a direct parallel of the idea of solving a differential equation by using a sequence of updates where each update is an Euler approximation. Just as in the solving of a differential equation, it is important that the successive steps are of appropriate size, and the determination of that size is not totally trivial.

There is a further parallel with what happens with differential equations. In the event of a large perturbation in the set $T$, the actual EMST graph for the perturbed problem may differ from that for the original problem. Trying to find it via a series of perturbations starting with the original problem (each one of which leaves unchanged the associated graph) will likely involve a sequence of $X$ matrices whose entries diverge somewhere within the perturbation process.  The differential equation analog to this is the phenomenon of a finite escape time.  

\section{Main Results}\label{sec:main_results}

In this section, we provide the main results that address Problem~\ref{prob:no_obstabcles}. We first present a cost function, which along with its partial derivatives, form the basis of our algorithm. We then prove nonsingularity of a key Hessian used in the algorithm, before providing the algorithm for computing $\delta s$ and show that it provides a first-order approximation of the new ESMT. 

\subsection{Cost Function and Associated Partial Derivatives}\label{ssec:cost_calcs}

To begin, with minor abuse of notation, let $\bar T$ and $\bar S$ be a specific realisation of an ESMT $\mathcal T$. Consider the cost function $J\colon \mathbb{R}^{2n} \times \mathbb{R}^{2k} \rightarrow \mathbb{R}$, associated with $\mathcal T$, which calculates the total distance of a tree as follows:
\begin{align}\label{eq:generalJ}
J(t,s) &= \sum_{(i,j) \in\mathcal{E}_{T}}\sqrt{(x_{t_i} - x_{t_j})^2 +(y_{t_i} - y_{t_j})^2} \nonumber \\
&+  \sum_{(i,j) \in\mathcal{E}_{T,S}}\sqrt{(x_{t_i} - x_{s_j})^2 +(y_{t_i} - y_{s_j})^2} \\
&+ \sum_{(i,j) \in\mathcal{E}_{S}}\sqrt{(x_{s_i} - x_{s_j})^2 +(y_{s_i} - y_{s_j})^2} . \nonumber
\end{align}
Because $\bar S$ is minimising for a prescribed $\bar T$, $\min_s J(\bar t,s)$ is achieved at $s=\bar s$.  As it turns out, we can obtain an approximate solution to the problem in Problem~\ref{prob:no_obstabcles} using \eqref{eq:generalJ} and associated partial derivatives and the minimization property. Before doing so, we first present calculations for two partial derivatives of $J$ that will be needed: $\frac{\partial^2 J}{\partial s^2}$ (which we term the Hessian) and $\frac{\partial^2 J}{\partial t\partial s}$.

We begin with the Hessian.
First, we compute $\frac{\partial J}{\partial s} \in \mathbb R^{2k}$ as
\begin{equation}
    \frac{\partial J}{\partial s} = \begin{bmatrix}(\frac{\partial J}{\partial p_{s_1}})^\top & (\frac{\partial J}{\partial p_{s_2}})^\top & \cdots & (\frac{\partial J}{\partial p_{s_k}})^\top \end{bmatrix}^\top.
\end{equation}
which consist of the $2$-vectors
\begin{equation}\label{eq:partialJs_general}
    \frac{\partial J}{\partial p_{s_i}} = \sum_{(j,i)\in\mathcal E_{T,S}} \frac{p_{s_i} - p_{t_j}}{\Vert p_{t_j} - p_{s_i}\Vert} + \sum_{(\ell, i)\in\mathcal E_{S} } \frac{p_{s_i} - p_{s_{\ell}}}{\Vert p_{s_{\ell}} - p_{s_i}\Vert},
\end{equation}
for all $s_i \in S$. Note that we use different indices $j$ and $\ell$ for terminals and other Steiner nodes for clarity. Due to Property~\ref{degreecon}, there are precisely three summands in \eqref{eq:partialJs_general}.

Next, we calculate $\frac{\partial^2 J}{\partial s^2} = \frac{\partial}{\partial s} \left(\frac{\partial J}{\partial s}\right) \in \mathbb R^{2k\times 2k}$ to be
\begin{equation}\label{eq:partial2Js2_general}
    \frac{\partial^2 J}{\partial s^2} = \begin{bmatrix}
        J_{11} & J_{12} & \cdots & J_{1k} \\
        J_{21} & J_{22} & \cdots & J_{2k} \\
        \vdots & \vdots & \ddots & \vdots \\
        J_{k1} & J_{k2} & \cdots & J_{kk}.
    \end{bmatrix} 
\end{equation}
The block matrices $J_{i\ell}$ in \eqref{eq:partial2Js2_general}, for $i,\ell\in \{1,2,\hdots,k\}$, take on specific forms. The diagonal blocks are
\begin{align}\label{eq:Jii_block}
    J_{ii} & = \sum_{(j,i)\in\mathcal E_{T,S}} \frac{1}{\Vert p_{t_j} - p_{s_i}\Vert} \left(I_2 - \frac{(p_{s_i} - p_{t_j})(p_{s_i} - p_{t_j})^\top}{\Vert p_{t_j} - p_{s_i} \Vert^2}\right) \nonumber \\
    & \; + \sum_{(\ell,i)\in\mathcal E_{S}} \frac{1}{\Vert p_{s_\ell} - p_{s_i}\Vert} \left(I_2 - \frac{(p_{s_i} - p_{s_\ell})(p_{s_i} - p_{s_\ell})^\top}{\Vert p_{s_\ell} - p_{s_i} \Vert^2}\right),
\end{align}
where $I_2$ is the $2$-dimensional identity matrix. Due to Property~\ref{degreecon}, there are precisely three summands in \eqref{eq:Jii_block}. For $i\neq \ell$, the off-diagonal block $J_{i\ell}$ is the null matrix if $s_i$ is not connected to $s_\ell$ (i.e. $(i,\ell) \notin \mathcal E_{S}$), then $J_{i\ell} = 0_{2\times 2}$. This is because \eqref{eq:partialJs_general} involves summing over just those edges involving $s_i$. However, if $s_i$ and $s_\ell$ are connected, then
\begin{equation}\label{eq:Jij_block}
    J_{i\ell} = -\frac{1}{\Vert p_{s_\ell} - p_{s_i}\Vert} \left(I_2 - \frac{(p_{s_i}-p_{s_\ell})(p_{s_i}-p_{s_\ell})^\top}{\Vert p_{s_\ell} - p_{s_i}\Vert^2}\right).
\end{equation}

Now, we compute $\frac{\partial^2 J}{\partial t \partial s}\in \mathbb R^{2k\times 2n}$. Similarly to above, we can write this as a block matrix of the form
\begin{equation} \label{eq:partial2Jts}
    \frac{\partial^2 J}{\partial t \partial s} = \begin{bmatrix}
        \tilde{J}_{11} & \tilde{J}_{12} & \cdots & \tilde{J}_{1n} \\
        \tilde{J}_{21} & \tilde{J}_{22} & \cdots & \tilde{J}_{2n} \\
        \vdots & \vdots & \ddots & \vdots \\
        \tilde{J}_{k1} & \tilde{J}_{k2} & \cdots & \tilde{J}_{kn}
    \end{bmatrix},
\end{equation}
with the block $ \tilde J_{ij}=\frac{\partial^2 J}{\partial s_i\partial t_j}$ evaluating to be $\tilde J_{ij} = 0_{2\times 2}$ if $t_j$ and $s_i$ are not connected by an edge, and
\begin{equation}
    \tilde{J}_{ij} = -\frac{1}{\Vert p_{t_j} - p_{s_i} \Vert} \left(I_2 - \frac{(p_{s_i} - p_{t_j})(p_{s_i} - p_{t_j})^\top}{\Vert p_{t_j} - p_{s_i}\Vert^2}\right)
\end{equation}
if $t_j$ and $s_i$ are connected. Note that Property~\ref{degreecon} implies that each block row in \eqref{eq:partial2Jts} has at most three nonzero blocks.

Finally, we present a way to decompose the Hessian into summands that are indexed over the edges of $\mathcal T$, and our subsequent proof of the invertibility of the Hessian relies on properties of the summands. Observe that
\begin{equation}\label{eq:KLsum}
     \frac{\partial^2J}{\partial s^2}=\sum_{(j,i)\in\mathcal E_{T,S}}K_{ji}+\sum_{(m,\ell)\in\mathcal E_S}L_{m\ell},
 \end{equation}
and each summand corresponds to an edge in either $\mathcal E_{T,S}$ or $\mathcal E_S$. More precisely, the $K_{ji}$ terms correspond to edges between a terminal~$t_j$ and a Steiner point $s_i$, while the $L_{m\ell}$ terms correspond to edges between two Steiner points $s_m$ and $s_\ell$.

A summand $K_{ji}$, with $(j,i)\in \mathcal E_{T,S}$, is a $k\times k$ block matrix with each block a $2\times 2$ matrix, and the only nonzero block is the $ii$th block, which we denote by $K_{ji}^{(ii)}$. This block is 
 \begin{equation}\label{eq:K_ii}
     K^{(ii)}_{ji}=\frac{1}{\Vert p_{t_j} - p_{s_i}\Vert} \left(I_2 - \frac{(p_{s_i} - p_{t_j})(p_{s_i} - p_{t_j})^\top}{\Vert p_{t_j} - p_{s_i} \Vert^2}\right) .
 \end{equation}

A summand $L_{m\ell}$ with $(m,\ell)\in\mathcal E_S$ is a $k\times k$ block matrix with each block a $2\times 2$ matrix, and with the only nonzero blocks being the $mm, m\ell, \ell m$th and $\ell\ell$th blocks, denoted by $L^{(mm)}_{m\ell}, L^{(m\ell)}_{m\ell},L^{(\ell m)}_{m\ell}$ and $L^{(\ell\ell)}_{m\ell}$. In fact, if we define
\begin{equation}\label{eq:A_ml}
    A_{m\ell} = \frac{1}{\Vert p_{s_\ell} - p_{s_m}\Vert} \left(I_2 - \frac{(p_{s_m} - p_{s_\ell})(p_{s_m} - p_{s_\ell})^\top}{\Vert p_{s_\ell} - p_{s_m} \Vert^2}\right),
\end{equation}
then $L^{(mm)}_{m\ell} = L^{(\ell\ell)}_{m\ell} = A_{m\ell}$, and $L^{(m\ell)}_{m\ell} = L^{(\ell m)}_{m\ell} = -A_{m\ell}$.
Thus, $L_{m\ell}$ is a $k\times k$ block matrix which, after permutation operations are applied to both the rows and columns, is the direct sum of a $2\times 2$ matrix and a zero matrix of appropriate size. More explicitly, with $\Sigma_{m\ell}$ being the associated permutation matrix, we have
\begin{equation}\label{eq:L_ml_permuted}
    \Sigma_{m\ell} L_{m\ell}\Sigma_{m\ell}^{\top}=\begin{bmatrix}L^{(mm)}_{m\ell}&L^{(m\ell)}_{m\ell}&0&\dots&0\\
    L^{(\ell m)}_{m\ell}&L^{(\ell\ell)}_{m\ell}&0&\dots&0\\0&0&0&\dots&0\\
    \vdots&\vdots&\vdots& \ddots &\vdots\\
    0&0&0&\dots&0
    \end{bmatrix},
\end{equation}

\subsection{First-Order Algorithm}

Before we present the update algorithm, we demonstrate that the Hessian in \eqref{eq:partial2Js2_general}, evaluated at a specific $\bar T$ and $\bar S$ associated with an ESMT $\mathcal T$, is invertible. This property is key for the main approximation algorithm. In order to establish invertibility of the Hessian, we record a key property of the block matrices that form it.

\begin{lem}\label{lem:building_block}
    Let $a,b\in \mathbb{R}^2\setminus \{ 0 \}$. A matrix of the form 
\[
\frac{1}{\Vert b \Vert} \left ( I_2 - \frac{b b^\top}{\Vert b \Vert ^2} \right )
\]
is non-negative definite and has rank $1$; its nullspace consists of all scalar multiples of $b$. If $a$ and $b$ are linearly independent, then 
\[
\frac{1}{\Vert a \Vert} \left ( I_2 - \frac{a a^\top}{\Vert a \Vert ^2} \right )  +
\frac{1}{\Vert b \Vert} \left ( I_2 - \frac{b b^\top}{\Vert b \Vert ^2} \right )
\]
is positive definite.
\end{lem}
\begin{pf}
    The statement regarding the first matrix can be directly obtained from Fact~$3.13.4$, $3.13.7$ and~$3.13.9$ in \citep{bernstein2009matrix}, noting that $I_2 - \frac{b b^\top}{\Vert b \Vert ^2}$ is a projection matrix of rank $n-1$. The statement regarding the second matrix is established as follows. Since both summands are nonnegative definite, any $v$ that is a nullvector must be, separately, a nullvector of each summand. This means that $v$ needs to be a multiple of $a$ and a multiple of $b$, but the intersection of $\text{span}\{a\}$ and $\text{span}\{ b\}$ contains the zero vector only, as they are not linearly independent. \hfill $\qed$
\end{pf}

The invertibility of the Hessian is presented in the following result.
\begin{thm}\label{theorem:nonsingular}
    Let $\bar T$ be a given set of terminal nodes with position vector $\bar t$, and $\bar S$ be a set of Steiner nodes with position vector $\bar s$, which form an ESMT $\mathcal T$. The matrix 
    \[\frac{\partial^2J}{\partial s^2} (\bar t, \bar s),
    \]
    i.e. the Hessian as defined in~\eqref{eq:partial2Js2_general} and evaluated at $(\bar t, \bar s)$, is a nonsingular matrix. 
\end{thm}
\begin{pf}
    We will exploit the expression of the Hessian as given in \eqref{eq:KLsum}, and present the proof in two parts. In Part~a) we show that each summand in \eqref{eq:KLsum} is itself a non-negative definite symmetric matrix, and in Part~b) we show that the intersection of the nullspaces of all summands is trivial. Note that while the expressions in Section~\ref{ssec:cost_calcs} are for the Hessian evaluated for an arbitrary $t,s$, in this proof, all expressions are evaluated at $\bar t, \bar s$, i.e., for the specific ESMT in question.

\textit{Part a)} Consider $K_{ji}$, with $(j,i)\in \mathcal E_{T,S}$, as described below \eqref{eq:KLsum}. It is a block matrix in which the only nonzero block is the $ii$th block given in \eqref{eq:K_ii}, and this block is nonnegative definite by Lemma~\ref{lem:building_block}. It follows that $K_{ji}$ is a non-negative definite symmetric matrix. Therefore, the first term on the right-hand-side of \eqref{eq:KLsum} is a nonnegative definite matrix, as it is the sum of a set of nonnegative definite matrices.

Next, consider the summand $L_{m\ell}$ with $(m,\ell)\in\mathcal E_S$. As explained above, a permutation can be applied so that the summand reads as in \eqref{eq:L_ml_permuted}, which is a matrix of all zeros except the $4\times 4$ matrix on the upper diagonal:
\begin{align}\label{eq:Ldetail}
 \begin{bmatrix}
L^{(mm)}_{m\ell}&L^{(m\ell)}_{m\ell}\\
L^{(\ell m)}_{m\ell}&L^{(\ell\ell)}_{m\ell}   
\end{bmatrix} = \begin{bmatrix}
    I_2\\-I_2\end{bmatrix}
    A_{m\ell}
    \begin{bmatrix}
     I_2&-I_2
     \end{bmatrix}
\end{align}
The fact that this matrix is non-negative definite is obtained from Lemma~\ref{lem:building_block} and and Fact~$8.11.1$ in \citep{bernstein2009matrix}. Hence we have established that the Hessian is non-negative definite, being a sum of non-negative definite matrices. 

\textit{Part b)} For the original Steiner minimum tree $\mathcal{T} = (T \cup S,\mathcal E = \mathcal E_T  \cup \mathcal E_{T,S}  \cup \mathcal E_S)$, consider the graph induced by $S$. This can be written as $\mathcal{T}^\prime =  (S, \mathcal E_S)$, which defines a forest, a graph in which every connected component is a tree consisting of only Steiner points. Renumber the Steiner nodes so that those in the first connected component are the first occurring; then include those in the second connected component, and so on. If there are $d$ connected components in ${\mathcal T}^\prime$, then with this renumbering, the Hessian in~\eqref{eq:KLsum} can be expressed as
\begin{equation}
    \frac{\partial^2J}{\partial s^2} (\bar t, \bar s) = \begin{bmatrix}
        H_{1} & 0 & \cdots & 0 \\
        0 & H_2 & \cdots & 0 \\
        \vdots & \vdots & \ddots & \vdots \\
        0 & 0 & \cdots & H_d
    \end{bmatrix}.
\end{equation}
That is, the Hessian is a block diagonal matrix, where the matrix $H_r$ is the associated with the $r$th connected component in $T^\prime$ (the size of the $H_r$ can differ depending on the number of nodes in the connected component of the graph).

Now, consider the first connected component in $\mathcal{T}^\prime$, which is in fact a connected tree of $q\geq 1$ Steiner nodes, with nodes indexed as $s_1, \hdots, s_q$. If $[a_1^{\top}\;\dots a_k^{\top}]^{\top}$ is a nullvector of the Hessian, with each $a_i$ being a $2$-vector, then the vector $[a_1^\top\;\hdots\;a_{q}^\top\;0\;0\hdots\;0]^\top$ must be a nullvector of the diagonal block associated with the first connected component. We will show first that i) $a_1 = \hdots = a_q$ (this is trivial in the case of $q=1$), and then ii) that $a_1 = 0$, which proves the block  matrix associated with the first connected component is positive definite. In each connected component of $\mathcal{T}^\prime$, there must be at least one Steiner node that is, in the original graph, connected to at least two terminals. This is because the connected component is a tree of just Steiner nodes, and therefore in $\mathcal{T}^\prime$, there exists at least one Steiner node with at most degree one (zero if it is an isolated Steiner node). Clearly, based on Property~\ref{degreecon}, this Steiner node must be connected to at least two terminals in the original graph. Suppose without the loss of generality that in this connected component, this particular Steiner node is $s_1$ and it is connected to terminals $j_1$ and $j_2$. 

To prove that $a_1  = \hdots = a_q$ in the case where $s_1$ is not isolated, without loss of generality, let $s_1$ be connected to $s_2$. The relevant summands in \eqref{eq:KLsum} to recall are those $L_{ij}$ associated with edges $(s_i, s_j) \in \mathcal E_S$, such that $i,j = 1,\hdots,q$. 
In particular, the first summand is 
\begin{equation}\label{eq:L12}
    L_{12}=\begin{bmatrix}
      A_{12}&-A_{12}&\dots& 0\\
 -A_{12}&A_{12} &\dots&0\\
 0&0&\dots&0\\
 \vdots&\vdots&&\vdots
    \end{bmatrix},
\end{equation}
where $A_{12}$ is as given in \eqref{eq:A_ml} with $s_m = \bar s_1$ and $s_\ell = \bar s_2$. From the fact that the Hessian is the sum of nonnegative definite matrices, it follows that if $a = [a_1^{\top}\;a_2^{\top}\dots a^{\top}_k]^{\top}$ is a nullvector of the Hessian, it is necessarily a nullvector of $L_{12}$.  If $a$ is a nullvector, then \eqref{eq:L12} implies $a_2=a_1$. Following a similar approach, we can leverage the expression of $L_{s_is_j}$ to show that $a_i = a_j$ for all $(i,j)\in \mathcal E_S$ and $i,j =1,2,\hdots, q$. Thus, we have $a_1 = a_2 = \hdots = a_q$.

Now we show that $a_1 = 0$ when evaluated at $\bar t, \bar s$, and therefore each $a_1, \hdots, a_q$ is the zero vector. The first sum on the right-hand-side of \eqref{eq:KLsum} contains $K_{j_11}$ and $K_{j_21}$, which are both nonzero in the same $2\times 2$ block. 
More precisely, define
\begin{align}
    & B_{1}  = \frac{1}{\Vert p_{\bar t_{j_1}} - p_{\bar s_1}\Vert} \left(I_2 - \frac{(p_{\bar s_1} - p_{\bar t_{j_1}})(p_{\bar s_1} - p_{\bar t_{j_1}})^\top}{\Vert p_{\bar t_{j_1}} - p_{\bar s_1} \Vert^2}\right) \nonumber \\
    & \quad + \frac{1}{\Vert p_{\bar t_{j_2}} - p_{\bar s_1}\Vert} \left(I_2 - \frac{(p_{\bar s_1} - p_{\bar t_{j_2}})(p_{\bar s_1} - p_{\bar t_{j_2}})^\top}{\Vert p_{\bar t_{j_2}} - p_{\bar s_1} \Vert^2}\right)
\end{align}
to obtain
\begin{equation}\label{eq:Kj1_big}
    \sum_{p=1,2} K_{j_p1}=\begin{bmatrix}
        B_1&0&\dots&0\\
        0&0&\dots&0\\
        \vdots&\vdots&&\vdots\\
        0&0&\dots&0
        \end{bmatrix}.
\end{equation}
Any nullvector of $(\sum_{p=1,2} K_{j_p1})$ satisfies
\begin{equation}
   \left ( \sum_{p=1,2} K_{j_p1} \right ) \cdot 
\begin{bmatrix}
a_1^{\top} \\ 0 \\
\vdots \\0
\end{bmatrix} = 
\begin{bmatrix}
0 \\
0\\
\vdots \\0
\end{bmatrix}.
\end{equation}
However, Property~\ref{anglecon} states that $p_{\bar s_1} - p_{\bar t_{j_1}}$ and $p_{\bar s_1} - p_{\bar t_{j_2}}$ are 120 degrees apart and thus linearly independent. Using Lemma~\ref{lem:building_block}, we conclude that $B_1$ is positive definite and $a_1 = 0$. This means that all entries $a_1=\hdots =  a_q$ corresponding to the first connected component are zero, and the block diagonal matrix corresponding to the first connected component is positive definite. 

As noted in the first paragraph of \textit{Part b}, the Hessian can be expressed as a block diagonal matrix where each block corresponds to a connected component. Since we can always find a reordering of nodes so that any connected component of $\mathcal T^\prime$ can be the first component, proving the first block is positive definite immediately proves that each block matrix corresponding to a connected component is positive definite, which completes the proof.
\hfill $\qed$
\end{pf}

We are in a position to present the main algorithm of this paper. 

\begin{thm}[First-order approximation theorem]\label{thm:algorithm} 
\textnormal{}\\ 
 Consider Problem~\ref{prob:no_obstabcles} with $\bar T$ as the specific set of terminal nodes, and $\bar S$ a specific set of Steiner nodes such that $\bar{\mathcal T} = (\bar T\cup \bar S, \mathcal E)$ is an ESMT. Suppose that $\delta t$ is sufficiently small such that the number of Steiner points, $\vert \bar S\vert$, is unchanged and the edge set of $\bar{\mathcal T}$ and ${\bar{\mathcal{T}}}_{\bar t + \delta t}$ are the same. 
    Then, the perturbation vector $\delta s$ which solves Problem~\ref{prob:no_obstabcles} is given by 
    \begin{align}\label{eq:changeS}
\delta s \approx - \left ( \frac{\partial ^2 J}{\partial s^2} (\bar t, \bar s) \right )^{-1}   \frac{\partial ^2 J}{\partial t \partial s} (\bar t, \bar s) \cdot \delta t
\end{align}
to the first order.
\end{thm}
\bigskip

Before we present the proof of Theorem~\ref{thm:algorithm}, we make note of some aspects of the theorem. A key assumption in Theorem~\ref{thm:algorithm} is that the change $\delta t$ is sufficiently small so that the number $k$ of Steiner points and the edges $\mathcal E$ in the resulting minimal Steiner tree remains the same. This can be a strong assumption considering even small changes may result in a ``reconfiguration" of the tree in which either the number of Steiner nodes or the edges change. As we will show with numerical examples in Section~\ref{sec:examples}, our approach works well when there are no reconfigurations, even if the perturbation is substantial and large. This is achieved by replacing a single large perturbation by a sequence of small perturbations whose aggregate corresponds to the large perturbation.

\begin{pf}[Theorem~\ref{thm:algorithm}]
The definition of \eqref{eq:generalJ} implies that for a given $t = \bar t$, the location vector $s = \bar s$ must be a minimiser of $J(s,\bar t)$. Therefore,
\begin{equation}\label{eq:partialJs_0}
    \frac{\partial J}{\partial s} ( \bar{t}, \bar{s}) = 0.
\end{equation}
Moreover, if $\bar s + \delta s$ is the vector for the new Steiner point locations associated with the perturbed terminal locations $\bar t + \delta t$, then equally there must hold:
\begin{equation}
    \frac{\partial J( \bar{t} + \delta \bar{t}, \bar{s} + \delta \bar{s})}{\partial s} = 0.
\end{equation}
We shall make use of this to establish \eqref{eq:changeS}.

\textit{First-Order Condition}: Let $\bar t$ be fixed but arbitrary. We consider $s =  s (\bar t)$ as a function of $\bar t$, and suppose $\bar s$ is such that it minimises $J$ for a given $\bar t$. Since $\frac{\partial J}{\partial s} (\bar{t}, \bar{s} (\bar t )) = 0$ must hold for any value of $\bar t$, we differentiate both sides of \eqref{eq:partialJs_0} with respect to $t$ to obtain:
\begin{equation}
    D_t \left ( \frac{\partial J}{ \partial s} (\bar{t}, \bar{s}(\bar t)) \right )    = 0,
 \end{equation}
 where $D_t(f)$ is the total derivative of a vector-valued function $f$ with respect to the vector $t$. Using the chain rule yields
\begin{align}
       D_t \left ( \frac{\partial J}{ \partial s} (\bar{t}, \bar{s} (\bar t) ) \right ) &= 0\\
  \Leftrightarrow  \ \ \ \   \frac{\partial ^2 J}{\partial s^2}(\bar{t}, \bar{s} ) \ D_t(s) + \frac{\partial ^2 J}{\partial t \partial s  }(\bar{t}, \bar{s} )  &= 0.
\end{align}
Theorem~\ref{theorem:nonsingular} establishes that $\frac{\partial ^2 J}{\partial s^2} (\bar t, \bar s)$ is invertible, and rearranging thus yields
\begin{align}
  D_t (s)  = - \left ( \frac{\partial ^2 J}{\partial s^2}(\bar{t}, \bar{s} ) \right )^{-1} \cdot   \frac{\partial ^2 J}{\partial t \partial s} (\bar{t}, \bar{s} ) .
\end{align}
\textit{Estimating the change in the Steiner points:}  With the above calculations we can now approximate the change in the Steiner point positions assuming the number of Steiner points and edge connections stay the same. For a small change in the terminal node positions $\delta t$, the corresponding change in the Steiner points $\delta s$ is approximately
\[
\delta s  \approx \left. D_t (s) \right\rvert_{(\bar t , \bar s)} \cdot \delta t ,
\]
which then leads to \eqref{eq:changeS} as claimed. \hfill $\qed$
\end{pf}

 Although the resulting network does not, in general, constitute an exact ESMT, it remains a close approximation of the true minimal structure. In particular, the increase in total length is small.

\section{Numerical Examples}\label{sec:examples}

In this section, we present several numerical examples in order to showcase the effectiveness of the proposed method, as scenarios in which it may struggle. All numerical values are reported to three decimal places.

\subsection{Example 1}

We begin with a three-terminal example, for which one can analytically prove that when the terminals  are positioned such that they form a triangle with interior angles of  less than $120^\circ$, there is precisely one Steiner point that is connected to each of the terminals and satisfies Properties~\ref{degreecon} and \ref{anglecon}. Suppose that the terminals are $T = \{t_1, t_2, t_3\}$ with $  p_{t_1} = [0,0]^\top, p_{t_2} = [1,0]^\top$, and $p_{t_3} = [0,1]^\top$. For these terminals, the Steiner point in the Steiner minimum tree has the coordinates $\bar s \approx [0.211,0.211]^\top$.

\begin{figure*}[!htbp]
\centering
\begin{minipage}{0.3\textwidth}
    \centering
    \includegraphics[width=\textwidth]{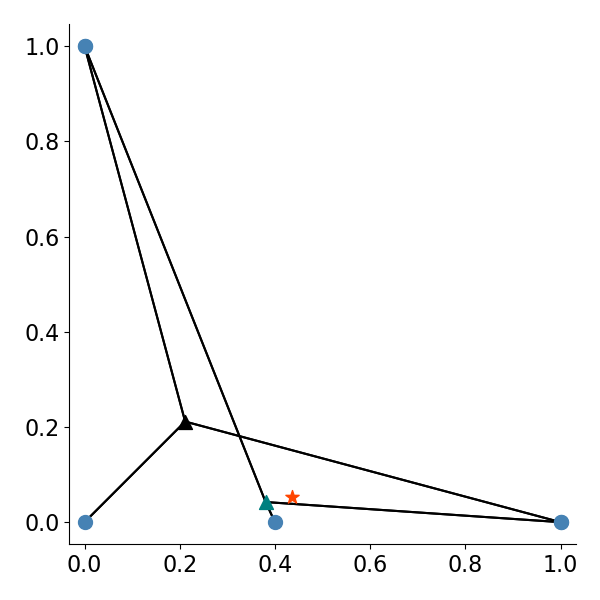}
    \subcaption{}
    \label{fig:shiftedSteiner_0.4_oneStep}
\end{minipage}%
\hspace{0.2\textwidth}
\begin{minipage}{0.3\textwidth}
    \centering
    \includegraphics[width=\textwidth]{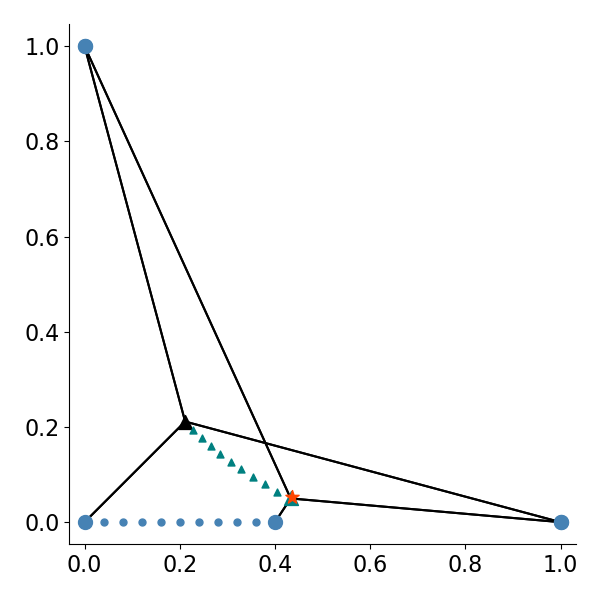}
    \subcaption{}
    \label{fig:shiftedSteiner_0.4}
\end{minipage}
\caption{Our approach for a 3-terminal problem with a large perturbation. Panel (a) shows a single-step calculation, while Panel (b) shows the same problem using an iterated stepwise shift with 10 steps. See Section~\ref{ssec:stepwise_example} for a description of the different nodes.}
\label{fig:shiftedSteiner}
\end{figure*}

Now, we consider a small displacement of $t_1$  so that $p_{t_1} + p_{\delta t_1} = [0.1, 0]^T$, and therefore $\delta t = [0.1, 0, 0, 0, 0, 0]^\top$. We calculate the new approximate Steiner point position using the algorithm given in Theorem~\ref{thm:algorithm}. 
The objective function in this case is $J\colon \mathbb{R}^6 \times  \mathbb{R}^2 \rightarrow \mathbb{R}$ with
\begin{align}
J (t,s) =  \sum_{j=1}^{3} \lVert p_{t_j} - p_s \rVert . 
\end{align}

For the calculations that follow, values will be rounded to three decimal places.
Using~\eqref{eq:partial2Js2_general} to \eqref{eq:Jij_block}, we get
\begin{align}
\frac{\partial^2 J}{\partial s^2} (\bar t, \bar s)
&=
\begin{bmatrix}
2.898 & -1.061 \\
-1.061 & 2.898
\end{bmatrix}
\end{align}

The mixed partial derivative $\frac{\partial^2 J}{\partial t \partial s} (\bar t, \bar s)$ is given by:
\begin{align}
  &\frac{\partial^2 J}{\partial t \partial s} (\bar t, \bar s)  \\
   &\approx
  \begin{pmatrix}
        -1.673 & 1.673 &  -0.082 & -0.306  &  -1.143& -0.306 \\
        1.673 &-1.673 & -0.306 & -1.143 & -0.306 & -0.082
  \end{pmatrix} \nonumber
 \end{align}
We now determine the approximate change in $\bar s $ when $\delta\bar t= [0.1 \  0 \ 0\ 0\ 0\ 0]^\top$:
\begin{align}\label{eq:calculateS}
\delta \bar s &\approx - \left ( \frac{\partial ^2 J}{\partial  s^2} (\bar t, \bar s) \right )^{-1} \cdot   \frac{\partial ^2 J}{\partial  t \partial  s} (\bar t, \bar s) \cdot \delta \bar t  \nonumber \\
    &\approx 
        \begin{pmatrix}
        0.423 & -0.423 & 0.077 & 0.289  & 0.500 & 0.134 \\
        -0.423 & 0.423 & 0.134  & 0.500 & 0.289 & 0.077
        \end{pmatrix}
        \cdot \delta \bar t
\end{align}

We obtain $\delta s = [0.0423, -0.0423]^\top$, and therefore 
\begin{align}
    \bar s + \delta \bar s 
    =
       \begin{pmatrix}
        0.254  \\
        0.169
    \end{pmatrix} 
\end{align}

For comparison, when we analytically calculate the position of the Steiner point directly from the new terminal node set, there results $\bar s + \delta \bar s= [0.257, 0.169]^\top$, illustrating the fact that our approximation is accurate.

\subsection{Example 2: Stepwise shift}\label{ssec:stepwise_example}

\begin{figure*}[!htbp]
\centering
\begin{minipage}{0.3\textwidth} 
    \centering
    \includegraphics[width=\textwidth]{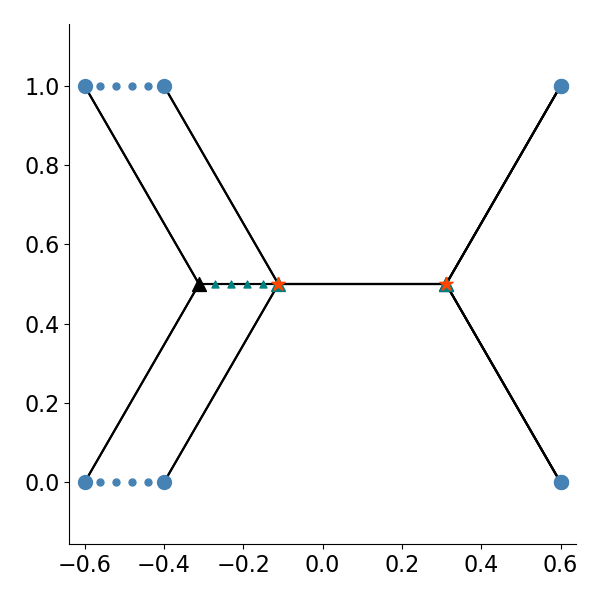}
    \subcaption{}
    \label{fig:fourPoints2}
\end{minipage}%
\hspace{0.2\textwidth} 
\begin{minipage}{0.3\textwidth} 
    \centering
    \includegraphics[width=\textwidth]{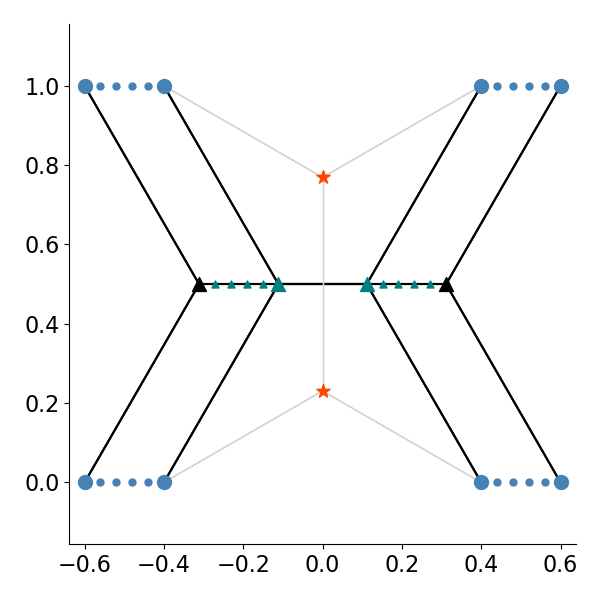}
    \subcaption{}
    \label{fig:fourPoints}
\end{minipage}
\caption{Four terminal example. Panel (a) shows that when the left terminals are shifted, the algorithm continues to perform well. Panel (b) shows that when the right terminals are also shifted, the algorithm breaks down, as seen by the substantial difference between the actual Steiner point location (red stars) and our approximated solution (blue triangles).}
\label{fig:fourPointsComparison}
\end{figure*}

For the same three-terminal setup in \textit{Example 1}, let us now consider a sizeable move of terminal $t_1$, with $\delta t = [0.4, 0, 0, 0, 0, 0]^\top$. The actual location of the new Steiner point for the shifted terminals $t'_1 = [0.4, 0]^\top$, $t'_2 = [1, 0]^\top$, $t'_3=[0, 1]^\top$ was $[0.437, 0.052]^\top$. Applying the algorithm in Theorem~\ref{thm:algorithm} yields the approximate Steiner point position for the perturbed terminal positions as $\bar s + \delta{\bar s} = [0.380, 0.042]^\top$. Figure~\ref{fig:shiftedSteiner_0.4_oneStep} illustrates the original Steiner tree and the perturbed tree with the approximate and true Steiner point positions. The three outer blue circles are the terminals, and the black triangle is the original Steiner point. The blue circle at $[0.4,0]$ is the shifted terminal $t_1$, and the red star indicates the actual new Steiner node location. The blue triangle is the approximate new Steiner node location obtained using our algorithm.

Suppose instead that we split the single $0.4$ perturbation of $x_{t_1}$ into a sequence of ten smaller perturbations of $0.04$. In this case, we use our proposed algorithm iteratively. In each iteration, we shift $\delta x_{t_1}$ by $0.04$  and obtain a new $\delta s$. Then, we update $\bar s \leftarrow \bar s + \delta s$ for the next iteration, and repeat. With this approach we arrive at the final approximate Steiner point having position $p_s = [0.433, 0.050]^\top$. Figure~\ref{fig:shiftedSteiner_0.4} shows the result of the iterative stepwise approach, with the smaller circle and triangle (partially overlapping in the figure) corresponding to the iteratively obtained positions of the terminal and the approximate Steiner point, very close to the true point.

\subsection{Limits of the approximation theorem}
As we have seen in the previous section, the stepwise application of the approximation theorem works well for an example where the perturbation in the terminals does not affect the connections in the Steiner minimum tree. In the following example, there are four terminals to be connected (see Fig.~\ref{fig:fourPoints}). When shifting the two left terminals towards the centre, the stepwise approximation performs well as shown in Figure~\ref{fig:fourPoints2}, as shown by the proximity of the blue triangles (approximate Steiner points) and the red stars (actual new Steiner points). Shifting the two terminals on the left toward the centre in five steps results in an updated position of the left Steiner point at $p_{s_1} = [-0.11132486, 0.5 ]^\top$, while the true minimising Steiner point location would be $p_{s_1} = [-0.11132487,  0.5 ]^\top$. However, if we additionally shift the right-side terminals towards the centre in the same way (see Fig.~\ref{fig:fourPoints}) , the actual Steiner minimum tree connects the terminals vertically instead of horizontally, as illustrated by the different location of the red stars (the actual new Steiner points). While the total lengths of the approximated and the minimal Steiner tree solution are still relatively close for small changes, the proposed stepwise method cannot find the new topology of the Steiner minimum tree solution. This is shown in Figure~\ref{fig:fourPoints}, where the gray lines show the edges of the actual Steiner minimum tree, while the approximated Steiner tree solution identified by our approach has black lines.

\section{Conclusion and future work}
In conclusion, this study explored the recovery of a Euclidean Steiner minimum tree following perturbations to terminal node positions using a first-order approximation theory. The mathematical groundwork laid here paves the way for future research into more complex scenarios involving perturbations, where the computational complexity of solving the problem from scratch provides a strong motivation to pursue our approach. In particular, we aim to extend our approach to consider the presence of soft obstacles, with a cost penalty for traversing, arising in pipeline networks or transmission lines through irregular terrain. As a first step, we consider a very soft obstacle as a perturbation of the zero obstacle problem studied in this paper, and then dealing with increasing obstacle strength through a stepwise approach similar to Section~\ref{sec:examples}. Additional directions include consideration of hard obstacles which cannot be traversed, and trees in non-planar $2D$ surfaces or in $3D$ space.

\bibliography{ifacconf}             

\end{document}